\documentclass[11pt]{amsart}
\usepackage[english]{babel}
\usepackage{physics}
\usepackage{braket}
\usepackage{dsfont}
\usepackage{xcolor}
\usepackage{graphicx}
\usepackage{tikz}
\usetikzlibrary{matrix}

\usepackage[margin=1in]{geometry}
\usepackage[pagebackref, colorlinks = true, linkcolor = blue, urlcolor  = blue, citecolor = red]{hyperref}

\usepackage{amsfonts}
\usepackage{amsmath}
\usepackage{amssymb}
\usepackage{amsthm}

\usepackage{blkarray}

\usepackage{physics}

\usepackage{pgfplots}
\pgfplotsset{compat = 1.16}

\usepackage{float}

\newtheorem{definition}{Definition}
\newtheorem{theorem}{Theorem}
\newtheorem{lemma}{Lemma}

\DeclareMathSymbol{\shortminus}{\mathbin}{AMSa}{"39}

\DeclareMathOperator{\id}{id}
\DeclareMathOperator{\Span}{span}
\DeclareMathOperator{\spec}{spec}
\DeclareMathOperator{\transp}{transp}

\newcommand{\bigzero}{\mbox{\normalfont\Large 0}}

\usepackage{tikz}
\usepackage{xstring}
\newcommand{\permutationOperator}[1]{
    \node[site] (l1)                 {};
    \node[site] (l2) [below of = l1] {};
    \node[site] (l3) [below of = l2] {};
    \node[]     (m1) [right of = l1] {};
    \node[]     (m2) [right of = l2] {};
    \node[]     (m3) [right of = l3] {};
    \node[site] (r1) [right of = m1] {};
    \node[site] (r2) [right of = m2] {};
    \node[site] (r3) [right of = m3] {};
    \IfStrEqCase{#1}{
        {id}{
            \draw[-] (l1) to (r1);
            \draw[-] (l2) to (r2);
            \draw[-] (l3) to (r3);
        }
        {12}{
            \draw[-] (l1) to (r2);
            \draw[-] (l2) to (r1);
            \draw[-] (l3) to (r3);
        }
        {13}{
            \draw[-] (l1) to (r3);
            \draw[-] (l2) to (r2);
            \draw[-] (l3) to (r1);
        }
        {23}{
            \draw[-] (l1) to (r1);
            \draw[-] (l2) to (r3);
            \draw[-] (l3) to (r2);
        }
        {123}{
            \draw[-] (l1) to (r2);
            \draw[-] (l2) to (r3);
            \draw[-] (l3) to (r1);
        }
        {321}{
            \draw[-] (l1) to (r3);
            \draw[-] (l2) to (r1);
            \draw[-] (l3) to (r2);
        }
    }
}
\newcommand{\transposedPermutationOperator}[1]{
    \node[site] (l1)                 {};
    \node[site] (l2) [below of = l1] {};
    \node[site] (l3) [below of = l2] {};
    \node[]     (m1) [right of = l1] {};
    \node[]     (m2) [right of = l2] {};
    \node[]     (m3) [right of = l3] {};
    \node[site] (r1) [right of = m1] {};
    \node[site] (r2) [right of = m2] {};
    \node[site] (r3) [right of = m3] {};
    \IfStrEqCase{#1}{
        {id}{
            \draw[-] (l1) to                   (r1);
            \draw[-] (l2) to                   (r2);
            \draw[-] (l3) to                   (r3);
        }
        {12}{
            \draw[-] (l1) to [bend left = 90]  (l2);
            \draw[-] (r1) to [bend right = 90] (r2);
            \draw[-] (l3) to                   (r3);
        }
        {13}{
            \draw[-] (l1) to [bend left = 60]  (l3);
            \draw[-] (r1) to [bend right = 60] (r3);
            \draw[-] (l2) to                   (r2);
        }
        {23}{
            \draw[-] (l1) to                   (r1);
            \draw[-] (l2) to                   (r3);
            \draw[-] (l3) to                   (r2);
        }
        {123}{
            \draw[-] (l1) to [bend left = 90]  (l3);
            \draw[-] (r1) to [bend right = 90] (r2);
            \draw[-] (l2) to                   (r3);
        }
        {321}{
            \draw[-] (l1) to [bend left = 90]  (l2);
            \draw[-] (r1) to [bend right = 90] (r3);
            \draw[-] (l3) to                   (r2);
        }
    }
}
\newcommand{\basisVector}[1]{
    \node[site] (l1)                 {};
    \node[site] (l2) [below of = l1] {};
    \node[site] (l3) [below of = l2] {};
    \node[site] (r1) [right of = l1] {};
    \node[site] (r2) [right of = l2] {};
    \node[site] (r3) [right of = l3] {};
    \IfStrEqCase{#1}{
        {1}{
            \draw[-] (l1) to [bend left = 90] (l2);
            \draw[-] (l3) to                  (r3);
        }
        {2}{
            \draw[-] (l1) to [bend left = 90] (l3);
            \draw[-] (l2) to                  (r2);
        }
    }
}
\newcommand{\R}[1]{
    \resizebox{3em}{!}{
        \begin{tikzpicture}[line width = 2pt,
                            site/.style = {circle,
                                           fill = black!80!white,
                                           thick}]
            \permutationOperator{#1}
        \end{tikzpicture}
    }
}
\newcommand{\RT}[1]{
    \resizebox{3em}{!}{
        \begin{tikzpicture}[line width = 2pt,
                            xscale = 1.5,
                            site/.style = {circle,
                                           fill = black!80!white,
                                           thick}]
            \transposedPermutationOperator{#1}
        \end{tikzpicture}
    }
}
\newcommand{\smallRT}[1]{
    \raisebox{\dimexpr-.5\height+.5\depth+.8ex\relax}{
        \boxed{
            \hspace{-\fboxsep}
            \resizebox{!}{1.2em}{
                \begin{tikzpicture}[line width = 4pt,
                                    xscale = 1.5,
                                    site/.style = {draw = none,
                                                   fill = none}]
                    \transposedPermutationOperator{#1}
                \end{tikzpicture}
            }
            \hspace{-\fboxsep}
        }
    }
}
\newcommand{\V}[1]{
    \raisebox{\dimexpr-.5\height+.5\depth+.8ex\relax}{
        \boxed{
            \hspace{-\fboxsep}
            \resizebox{!}{1.2em}{
                \begin{tikzpicture}[line width = 4pt,
                                    xscale = 1.5,
                                    site/.style = {draw = none,
                                                   fill = none}]
                    \basisVector{#1}
                \end{tikzpicture}
            }
            \hspace{-\fboxsep}
        }
    }
}

\title{A geometrical description of the universal~$1 \to 2$ asymmetric quantum cloning region}

\author{Ion Nechita}
\email{nechita@irsamc.ups-tlse.fr}
\address{Laboratoire de Physique Th\'eorique, Universit\'e de Toulouse, CNRS, UPS, France}
\author{Cl\'ement Pellegrini} 
 \email{clement.pellegrini@math.univ-toulouse.fr
}
\address{Institut de Math\'ematiques, IMT, Universit\'e de Toulouse (UMR 5219), CNRS, UPS, 31062 Toulouse, Cedex 9, France}
\author{Denis Rochette}
\email{denis.rochette@math.univ-toulouse.fr}
\address{Institut de Math\'ematiques, IMT, Universit\'e de Toulouse (UMR 5219), CNRS, UPS, 31062 Toulouse, Cedex 9, France}

\date{\today}

\begin{document}
\begin{abstract}
We consider the problem of determining the achievable region of parameters for universal $1 \to 2$ asymmetric quantum cloning. Measuring the cloning performance with the figure of merit of singlet fraction, we show that the physical region is a union of ellipses in the plane. Equivalently, we characterize the parameter region of quantum state compatibility of two possibly different isotropic states, considering, for the first time, negative singlet fractions.
\end{abstract}

\maketitle

\tableofcontents

\section{Introduction}

The fact that quantum information cannot be copied is one of the most striking differences between quantum information theory and classical information theory. This fact, discovered in the early days of modern quantum information theory \cite{wootters1982single}, is a no-go result from a quantum communication perspective, being at the same time a cornerstone of cryptography protocols: the impossibility of non-cloning prevents a malicious eavesdropper from intercepting a message and copying it without disturbing the original.  Perfect cloning being possible only for families of perfectly distinguishable quantum states, one can naturally ask whether one can relax this very stringent requirement by asking for approximate clones of given qualities. This is the topic of \emph{universal asymmetric quantum cloning}, the topic of the current paper. 

The quantum cloning problem has received a lot of attention in the past three decades. Starting from the early pioneering work on universal quantum cloners \cite{buvzek1996quantum}, many authors studied the different cloning scenarios (symmetric vs.~asymmetric, qubit vs.~qudit, etc) \cite{cerf1998asymmetric,werner1998optimal,keyl1999optimal,cerf2000asymmetric,fiuravsek2005highly}. Two series of papers are concerned with the most general, asymmetric, $1 \to N$ quantum cloning problem: one from Kay and collaborators \cite{kay2012optimal,kay2014optimal,kay2016optimal}, and another one using techniques from group representation theory, by \'Cwikli\'nski, Horodecki, Mozrzymas, and Studzi\'nski \cite{cwiklinski2012region,studzinski2013commutant,studzinski2014group}. Importantly for us, Hashagen studied in \cite{hashagen2016universal} the $1 \to 2$ universal asymmetric case, focusing on different figures of merit; the techniques used in that work are based on previous results of Eggeling and Werner \cite{eggeling2001separability} and Vollbrecht and Werner \cite{vollbrecht2001entanglement} about the separability of symmetric states. 

In the current work, we focus on universal $1 \to 2$ asymmetric cloning, in the general case of qudits (quantum systems of arbitrary fixed dimension). We shall study this problem from the following perspective: given some target pair of figures of merit, does there exist a quantum channel from one copy of the system to two copies, having the desired qualities? We shall restrict ourselves to the fidelity figure of merit (and to the singlet fraction \cite{kay2009optimal}). We prove, starting from first principles, that the achievable fidelity region is a \emph{union of ellipses}, with the optimal one coming from a restricted class of quantum cloners. Although the shape of the fidelity region was derived previously (see, e.g.~\cite[Section 6]{hashagen2016universal}), our approach is self-contained and gives a geometrical intuition on the desired region. We also show that, in order to reach non-optimal points from the achievable region, one has to use the full generality of the symmetrized quantum cloners available to us, providing the range of achievable negative singlet fractions. We informally summarize our results here, see Theorem \ref{thm:restricted} and Figure \ref{FiguresMerit:1}, and Theorem \ref{thm:unrestricted} and Figure \ref{FiguresMerit:2} respectively. Briefly speaking, 
we shall consider quantum channels $T: \mathcal{M}_d \to \mathcal{M}^{\otimes 2}_{d}$ such that the marginals $T_i$ satisfy $T_i(\rho)=p_i\rho+(1-p_i)\frac{I_d}{d}$; in this setting we call the parameters $p_i$ the singlet fractions. We shall also study the associated quantum fidelity $f_i=F(\rho,T_i(\rho))$ which actually does not depend on the state $\rho$. We are interested in the set of the pairs  $(p_1, p_2)$ (resp.~$(f_1, f_2)$) which are physically realizable. 

\begin{theorem}
	The achievable set of quantum fidelities (resp.~singlet fractions) for the $1 \to 2$ asymmetric quantum cloning problem can be described by two real parameters $s,t$ as follows: 
	$$ \begin{cases}
        s &= \: \frac{d(f_1+f_2) - 2}{d-1}\\
        t &= \: \frac{d(f_1-f_2) }{d-1}
	\end{cases}
	\qquad \text{ and } \qquad
	\begin{cases}
		s &= \: p_1 + p_2\\
		t &= \: p_1 - p_2.
	\end{cases}$$
Let $a=\frac{1}{\sqrt{d^2-1}}$ and $b=\frac{1}{d^2-1}$. The figures of merit for the the optimal quantum cloners are given by the ellipse
$$ \frac{t^2}{a^2}  +\frac{\left(s - \frac{d^2-2}{d^2-1} \right)^2}{b^2} \leq d^2.$$

In the general case, with no restriction on the form of the cloner, the achievable set of figures of merit is a union of a 1-parameter family of ellipses indexed by $\lambda \in [0,d]$ described by
$$ \frac{t^2}{a^2} + \frac{\left(s - \frac{\lambda \, d - 2}{d^2 - 1}\right)^2}{b^2} \leq \lambda^2.$$
 The optimal cloners correspond to $\lambda = d$. 
\end{theorem}

Let us stress out that the novel contribution of the current work is threefold: 
\begin{itemize}
	\item we show that optimal quantum cloners can be constructed from a restricted set of permutation operators, which describe an ellipse in the set of achievable fidelities (or singlet fractions);
	\item allowing for fully general cloners, we describe the achievable set as a union of ellipses, painting a geometrical picture of $1 \to 2$ asymmetric quantum cloning;
	\item our approach is self-contained, starting from first principles: we use the symmetrization technique of \cite{eggeling2001separability,vollbrecht2001entanglement} and then study the set of linear combinations of permutation operators giving valid quantum channels using basic linear algebra.
\end{itemize}

The paper is organized as follows. In Section \ref{sec:asym-q-cloning}, we recall some basic definitions and set up the problem. Section \ref{sec:symmetrization} contains the symmetrization argument, reducing the problem to a linear algebra one. In Section \ref{sec:graphical} we briefly review the graphical formalism for tensor calculus, which will be of help in Section \ref{sec:math-clone}, where we study the properties of the permutation operators of interest. Sections \ref{sec:restricted} and \ref{sec:general} contain the main results of the paper, the description of the achievable region for $1 \to 2$ asymmetric quantum cloning. 

\section{Asymmetric Quantum Cloning}\label{sec:asym-q-cloning}
Let~$\mathcal{H} = \mathbb{C}^d$ be a finite~$d$ dimensional Hilbert space. We denote by~$\mathcal{M}_d$ the space of~$d \times d$ complex matrices and by~$\mathcal{U}_d$ the group of unitary matrices. A \emph{quantum state} on~$\mathcal{H}$ is described by a density matrix~$\rho$, which is a unit trace, positive semi-definite matrix. We consider the set of all density matrices
\begin{equation*}
    \mathcal{D}_d = \big{\{} \rho \in \mathcal{M}_d \: \big{|} \: \rho \geq 0 \text{ and } \Tr \rho = 1 \big{\}}.
\end{equation*}
The set of density matrices is a convex set, where the extremal points are called \emph{pure quantum states}. These are unit rank projections~$\ketbra{\psi}{\psi}$, for~$\psi \in \mathcal{H}$ and~$\norm{\psi} = 1$. A \emph{quantum channel}~$T : \mathcal{M}_d \to \mathcal{M}_{d^\prime}$ describes the evolution of a quantum state. These are the completely positive and trace preserving linear (\textsc{cptp}) maps from~$\mathcal{M}_d$ to~$\mathcal{M}_{d^\prime}$, and represent the most general transform of a quantum state \cite{nielsen2010quantum,watrous2018theory}. On a composite quantum systems (i.e., a tensor product of Hilbert spaces), the \emph{marginals} of a quantum channel~$T : \mathcal{M}_d \to \mathcal{M}^{\otimes N}_{d^\prime}$ are denoted~$T_i$ for~$i \in [N]$ and defined by the partial trace~$T_i : \mathcal{M}_d \to \mathcal{M}_{d^\prime}$
\begin{equation*}
    T_i(X) = \Tr_{\scriptscriptstyle [N] \backslash \{i\}} \big{[} T(X) \big{]}, \qquad \forall X \in \mathcal{M}_d.
\end{equation*}
A marginal of a quantum channel is also a quantum channel. This gives rise to the question of compatibility of quantum channels \cite{heinosaari2016invitation,hsieh2021quantum}.
\begin{definition}[Quantum channel compatibility]
    Let~$\mathcal{H}_{0}$ be a input Hilbert space, and consider also output spaces $\mathcal{H}_1, \ldots \mathcal H_N$, together with quantum channels $T_i : \mathcal B(\mathcal H_0) \to \mathcal B(\mathcal H_i)$. The \emph{quantum channel compatibility problem} consists determining whether there exists a global quantum channel~$T:\mathcal B(\mathcal H_0) \to \mathcal B(\mathcal H:=\otimes_{i=1}^N\mathcal H_i)$, compatible with all the~$T_i$'s, that is
    \begin{equation*}
        \Tr_{\scriptscriptstyle \mathcal{H} \backslash \mathcal{H}_i} \circ \: T = T_i, \qquad \forall i = 1, \ldots, N.
    \end{equation*}
\end{definition}

The perfect~$1 \to N$ quantum cloning problem consists in determining whether there exists a quantum channel~$T : \mathcal{M}_d \to \mathcal{M}^{\otimes N}_{d}$, called a \emph{quantum cloning map}, compatible with the identity, i.e. for all~$i \in [N]$ and~$\rho \in \mathcal{D}_d$,~$T_i(\rho) = \rho$. The \emph{no-cloning theorem} \cite{wootters1982single} states that such a quantum channel cannot exist. Even if perfect quantum cloning is impossible, it can be performed approximately. That is, each marginal~$T_i(\rho)$ of a quantum cloning map must be as close as possible to the quantum state~$\rho$. The \emph{quantum cloning problem} is the problem of determining a quantum cloning map with the marginals the ``closest'' to the quantum state~$\rho$. In this work, we shall consider two figures of merit for the ``closeness'' of marginals, the quantum fidelity and the singlet fraction. The \emph{quantum fidelity} is a very well known distinguishability measure in quantum information theory which is defined for a quantum state $\rho$ and a pure state $\sigma = \ketbra{\psi}{\psi}$ by
\begin{equation*}
    F(\rho , \sigma) = 
\langle \psi | \rho | \psi \rangle.
\end{equation*}
Note that the quantum fidelity is not a metric.  We have $F(\rho , \sigma) \in [0 , 1]$, with~$F(\rho , \sigma) = 1$ if an only if~$\rho = \sigma = \ketbra{\psi}{\psi}$. In this paper, we are going to describe the set of fidelities of the marginals of quantum cloners, both in the average and the worst case scenario. For $1 \to 2$ cloners, these fidelities (which we shall call figures of merit) describe a convex region in the plane, which we shall describe in Theorems \ref{thm:restricted} and \ref{thm:unrestricted}. Other figures of merit in relation to the quantum cloning have been considered \cite{hashagen2016universal}. 


This setting gives rise to two quantum cloning problems: the \emph{symmetric} quantum cloning problem where the figures of merit are the same on each marginal, and the \emph{asymmetric} quantum cloning problem where the figures of merit can be different. In this paper we consider the more general asymmetric quantum cloning problem. Let~$s=(s_1,\ldots,s_N) \in [0 , 1]^N$ be a weight vector and~$T$ a quantum cloning map, we denote by
\begin{align*}
    \textit{(worst)} && F_s(T) &= \sum^N_{i = 1} s_i \cdot \inf_{\rho \text{ pure}} F \big{(} \rho , T_i(\rho) \big{)} \\
    \textit{(average)} && \bar{F}_s(T) &= \sum^N_{i = 1} s_i \cdot \mathop{\mathbb{E}}_{\rho \text{ pure}} \Big{[} F \big{(} \rho , T_i(\rho) \big{)} \Big{]}
\end{align*}
the \emph{worst} and \emph{average} quantum fidelity of the marginals of $T$, weighted according to the vector $s$. Since perfect cloning is impossible, the~$s$ vector allows us to give weights to the different marginals. This leads us to the two following problems.
\begin{definition}[Asymmetric quantum cloning problems]
    Let~$s$ be a weight vector in~$[0 , 1]^N$. The worst (resp. average) asymmetric quantum cloning problem is the supremum over the quantum channels of the worst (resp. average) quantum fidelity, i.e.
    \begin{align*}
        \textit{(worst)} & \qquad\qquad\qquad \sup_{T \textsc{ cptp}} F_s(T)  \\
        \textit{(average)} & \qquad\qquad\qquad \sup_{T \textsc{ cptp}} \bar{F}_s(T)
    \end{align*}
\end{definition}

\section{Permutation operators}\label{sec:symmetrization}
Finding the best quantum cloning map for a given weight vector~$s$ in~$[0 , 1]^N$ is an optimization problem over the set of quantum channels. A natural direction to approach the problems is to restrict the set of admissible quantum cloning maps.
\begin{definition}[Symmetrized quantum channel]
    Let~$T: \mathcal{M}_d \to \mathcal{M}^{\otimes N}_d$ be a quantum channel. The \emph{symmetrized version} of~$T$ is the quantum channel~$\widetilde{T}$ defined for all~$\rho \in \mathcal{D}_d$ by
    \begin{equation*}
        \widetilde{T}(\rho) = \int_{\mathcal{U}_d} U^{\otimes N} \, T \big{(} U^* \rho \: U \big{)} \, {\big{(} U^* \big{)}}^{\otimes N} \mathrm{d}U,
    \end{equation*}
    where the integral is taken with respected to the normalized Haar measure on the unitary group~$\mathcal{U}_d$.
\end{definition}
An important consequence is that for all~$U$ in~$\mathcal{U}_d$,~$\widetilde{T}(U^* \cdot \, U) = {\big{(} U^* \big{)}}^{\otimes N} \; \widetilde{T}(\cdot) \; U^{\otimes N}$

The symmetrization operation above, commonly called \emph{twirling}, is used as a primitive in a multitude of protocol in quantum information theory \cite{werner1998optimal,keyl1999optimal}. Note that the marginals of a symmetrized quantum channel are the symmetrized versions of the the marginals
\begin{equation*}
    \widetilde{T}_i(\rho) = \int_{\mathcal{U}_d} U \, T_i \big{(} U^* \rho \: U \big{)} \, U^* \mathrm{d}U.
\end{equation*}
It is known that the quantum fidelity (in its general incarnation, defined for arbitrary mixed states)is jointly concave \cite[Exercise 9.19]{nielsen2010quantum}, that is given two families of quantum states~$(\rho_i)_{i \in [n]}$ and~$(\sigma_i)_{i \in [n]}$, and a probability distribution~$(\pi_i)_{i \in [n]}$
\begin{equation*}
    F \bigg{(} \sum^n_{i = 0} \pi_i \, \rho_i , \sum^n_{i = 0} \pi_i \sigma_i \bigg{)} \geq \sum^n_{i = 0} \pi_i \cdot F(\rho_i , \sigma_i).
\end{equation*}
For any quantum cloning map~$T$, we have thus~$F \big{(} \rho , \widetilde{T}_i(\rho) \big{)} \geq F \big{(} \rho , T_i(\rho) \big{)}$. Therefore the set of admissible quantum cloning maps can be restricted to the symmetrized quantum channels, since by doing so, fidelities can only increase. In the following we consider only symmetrized quantum channels.

The \emph{Choi matrix}~$C_T$ of a linear map~$T : \mathcal{M}_d \to \mathcal{M}_{d^{\prime}}$ is defined by \cite[Section 2.2.2]{watrous2018theory}
\begin{align*}
    C_T &= (\id_d \otimes \, T) \bigg{(} \sum^d_{i , j = 1} \ketbra{i}{j} \otimes \ketbra{i}{j} \bigg{)} \\
    &= d \: (\id_d \otimes \, T) \ketbra{\Omega}{\Omega},
\end{align*}
where~$\ket{\Omega}$ is  the maximally entangled state. One can recover the linear map~$T$ from the Choi matrix~$C_T$ by the formula
\begin{equation*}
    T(X) = \Tr_d \big{[} C_T (X^T \otimes I_{d^{\prime}}) \big{]}.
\end{equation*}
The Choi matrices are used to classify the \textsc{cptp} linear maps with the following theorem.
\begin{theorem}[\cite{choi1975completely}] \label{ChoiTheorem:1}
    A linear map~$T : \mathcal{M}_d \to \mathcal{M}_{d^{\prime}}$ is a quantum channel if and only if its Choi matrix~$C_T$ is positive semi-definite and~$\Tr_{d^{\prime}} \big{[} C_T \big{]} = I_d$.
\end{theorem}

For all~$U$ in~$\mathcal{U}_d$, we have~$(\bar{U} \otimes U) \ket{\Omega} = \ket{\Omega}$. Let~$\widetilde{T} : \mathcal{M}_d \to \mathcal{M}^{\otimes N}_d$ be a symmetrized quantum channel, then
\begin{align*}
    C_{\widetilde{T}} \big{(} \bar{U} \otimes U^{\otimes N} \big{)} &= \Big{[} \big{(} \id \otimes \, \widetilde{T} \big{)} \ketbra{\Omega}{\Omega} \Big{]} \big{(} \bar{U} \otimes U^{\otimes N} \big{)} \\
    &= \Big{[} \big{(} \id \otimes \, \widetilde{T} \big{)} \big{(} \bar{U} \otimes U \big{)} \ketbra{\Omega}{\Omega} \big{(} U^T \otimes U^* \big{)} \Big{]} \big{(} \bar{U} \otimes U^{\otimes N} \big{)} \\
    &= \big{(} \bar{U} \otimes U^{\otimes N} \big{)} \Big{[} \big{(} \id \otimes \, \widetilde{T} \big{)} \ketbra{\Omega}{\Omega} \Big{]} \Big{(} U^T \otimes {\big{(} U^* \big{)}}^{\otimes N} \Big{)} \big{(} \bar{U} \otimes U^{\otimes N} \big{)} \\
    &= \big{(} \bar{U} \otimes U^{\otimes N} \big{)} C_{\widetilde{T}},
\end{align*}
where the third equation comes from the following property: for all~$U$ in~$\mathcal{U}_d$,~$\widetilde{T}(U^* \cdot \, U) = {\big{(} U^* \big{)}}^{\otimes N} \; \widetilde{T}(\cdot) \; U^{\otimes N}$. Therefore~$\big{[} C_{\widetilde{T}} , \bar{U} \otimes U^{\otimes N} \big{]} = 0$, which is equivalent to~$\big{[} C^{\Gamma}_{\widetilde{T}} , U \otimes U^{\otimes N} \big{]} = 0$ where~${(\cdot)}^{\Gamma}$ is the partial transpose on the first~$\mathcal{D}_d$ space, i.e.~${(v_1 \otimes v_2 \otimes \cdots \otimes v_N)}^{\Gamma} = v^T_1 \otimes v_2 \otimes \cdots \otimes v_N$.
\begin{theorem}[Schur–Weyl duality \cite{weyl2016classical}]
    Let~$L : \mathcal{H}^{\otimes n} \to \mathcal{H}^{\otimes n}$ be a linear map such that~$L$ commutes with~$U^{\otimes n}$, for all~$U$ in~$\mathcal{U}_d$, then~$L$ is a linear combination of \emph{permutation operators}~$V(\pi)$
    \begin{equation*}
        L = \sum_{\pi \in \mathfrak{S}_n} \alpha_\pi \cdot V(\pi)
    \end{equation*}
    where~$\mathfrak{S}_n$ is the symmetric group on~$n$ elements, the coefficients $(\alpha_\pi)_{\pi \in \mathfrak{S}_n}$ are complex numbers, and~$V(\pi)$ is the representation of~$\mathfrak{S}_n$, acting on~$\mathcal{H}^{\otimes n}$, defined by
    \begin{equation*}
        V(\pi)(v_1 \otimes \cdots \otimes v_n) = v_{\pi^{\shortminus 1}(1)} \otimes \cdots \otimes v_{\pi^{\shortminus 1}(n)}
    \end{equation*}
\end{theorem}

Taking the partial transpose on the right hand side leads us to the following characterization the Choi matrices of the symmetrized quantum channels.
\begin{lemma} \label{SchurWeylLemma:1}
    The Choi matrix~$C_{\widetilde{T}}$ of a symmetrized quantum channel~$\widetilde{T} : \mathcal{M}_d \to \mathcal{M}^{\otimes N}_d$ is a linear combination of partially transposed permutation operators, i.e.
    \begin{equation*}
        C_{\widetilde{T}} = \sum_{\pi \in \mathfrak{S}_{N+1}} \alpha_\pi \cdot V^{\Gamma}(\pi),
    \end{equation*}
where the partial transposition operator acts on the first of the $N+1$ copies of $\mathcal M_d$. 
\end{lemma}

Note that conversely, a linear combination of partially transposed permutation operators is a Choi matrix of a quantum channel, and hence a quantum cloning map, if the conditions of the Choi's Theorem \ref{ChoiTheorem:1} hold. 

In the particular case of the marginals of~$\widetilde{T}$, the formula above is very simple: there  are only two permutations in $\mathfrak S_2$, and their partial transpositions correspond to the identity, resp.~the maximally depolarizing channels. Hence, the marginals of a cloning map $\tilde T$ are of the form
$$\widetilde{T}_i(\rho) = p_i \, \rho + (1 - p_i) \frac{I}{d},$$
for some~$(p_i)_{i \in [N]}$ that do not depend on~$\rho$. The quantum fidelity of the marginals on any pure quantum state~$\rho = \ketbra{\psi}{\psi}$ are
\begin{align*}
    F \big{(} \rho , \tilde T_i(\rho) \big{)} &= \bra{\psi} \tilde T_i(\rho) \ket{\psi} \\
    &= \Big{\langle} \psi \Big{|} p_i \, \rho + \frac{1 - p_i}{d} I \Big{|} \psi \Big{\rangle} \\
    &= p_i + \frac{1 - p_i}{d}:=f_i
\end{align*}
Since the quantum fidelity does not depend on the choice of state, it follows that in the case of the marginals of a cloning map, the average and worst quantum fidelity coincide. We have thus proven the following important result, reducing the harder problem of the worse case fidelity to the one of the simpler, average fidelity. 

\begin{theorem} \label{ProblemsEquivalenceTheorem:1}
    The worst and the average quantum cloning problems are equivalent: for all~$s$ in~$[0 , 1]^N$,~$\sup_T F_s(T) = \sup_T \bar{F}_s(T)$, where both supremums are taken over the set of quantum channels $T:\mathcal M_d \to \mathcal M_d^{\otimes N}$.
\end{theorem}

\section{Graphical calculus}\label{sec:graphical}
In order to compute the aforementioned average (or, equivalently, worst-case) fidelities for quantum cloners, we introduce a graphical calculus for tensors, derived from the Penrose graphical notation \cite{penrose1971applications}. Similar calculi have been developed more recently in the framework of tensor network states or categorical quantum information theory can be found in \cite{wood2015tensor,bridgeman2017hand,coecke2017picturing}. 

In this diagrammatic notation, tensors are represented by \emph{boxes} and \emph{wires}
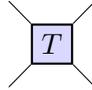
\begin{center}
    \begin{tikzpicture}[box/.style = {rectangle,
                                      draw = black,
                                      fill = blue!15!white,
                                      thick}]
        \node[box] (1)                      {$T$};
        \node[]    (2) [above left of = 1]  {};
        \node[]    (3) [above right of = 1] {};
        \node[]    (4) [below left of = 1]  {};
        \node[]    (5) [below right of = 1] {};
        \draw[-] (1.north west) to (2);
        \draw[-] (1.north east) to (3);
        \draw[-] (1.south west) to (4);
        \draw[-] (1.south east) to (5);
    \end{tikzpicture}
\end{center}
More specifically, the wires are labeled by \emph{indices}, such that a box represents the value of the tensor at the given indices
\begin{center}
 $T_{ij} = $ \begin{tikzpicture}[baseline = {([yshift = -1.5ex] current bounding box.center)},
                                 box/.style = {rectangle,
                                      draw = black,
                                      fill = blue!15!white,
                                      thick}]
        \node[box] (1)                {$T$};
        \node[]    (2) [left of = 1]  {};
        \node[]    (3) [right of = 1] {};
        \draw[-] (1) to node [midway, above] {$i$} (2);
        \draw[-] (1) to node [midway, above] {$j$} (3);
    \end{tikzpicture}   
\end{center}
In particular, a vector is a box with only~$1$ wire pointing to the left, and labeled by the index of the coordinate. A dual vector has its wire pointing to the right
\begin{center}
    \hfill
    $v_i =$ \begin{tikzpicture}[baseline = {([yshift = -1.5ex] current bounding box.center)},
                                box/.style = {rectangle,
                                      draw = black,
                                      fill = blue!15!white,
                                      thick}]
        \node[box] (1)               {$v$};
        \node[]    (2) [left of = 1] {};
        \draw[-] (1) to node [midway, above] {$i$} (2);
    \end{tikzpicture}
    \hfill
    $\bar v_i=$ \begin{tikzpicture}[baseline = {([yshift = -1.5ex] current bounding box.center)},
                                    box/.style = {rectangle,
                                      draw = black,
                                      fill = blue!15!white,
                                      thick}]
        \node[box] (1)                {$v^*$};
        \node[]    (2) [right of = 1] {};
        \draw[-] (1) to node [midway, above] {$i$} (2);
    \end{tikzpicture}
    \hfill\null
\end{center}
Note that these left and right directions are just a convention corresponding to the usual \emph{right-to-left} composition in linear algebra.

The tensor diagrams can be combined in two ways: the \emph{tensor product} and the \emph{tensor contraction}. The tensor product combines two diagrams vertically or horizontally
\begin{center}
    \hfill
    \begin{tikzpicture}[baseline = {([yshift = -0.5ex] current bounding box.center)},
                        box/.style = {rectangle,
                                      draw = black,
                                      fill = blue!15!white,
                                      thick}]
        \node[box] (1)                               {$A \otimes B$};
        \node[]    (2) [left of = 1, yshift = 1ex]   {};
        \node[]    (3) [left of = 1, yshift = -1ex]  {};
        \node[]    (4) [right of = 1, yshift = 1ex]  {};
        \node[]    (5) [right of = 1, yshift = -1ex] {};
        \draw[-] ([yshift = 1ex]1.west)  to (2);
        \draw[-] ([yshift = -1ex]1.west) to (3);
        \draw[-] ([yshift = 1ex]1.east)  to (4);
        \draw[-] ([yshift = -1ex]1.east) to (5);
    \end{tikzpicture}
    $=$
    \begin{tikzpicture}[baseline = {([yshift = -0.5ex] current bounding box.center)},
                        box/.style = {rectangle,
                                      draw = black,
                                      fill = blue!15!white,
                                      thick}]
        \node[box] (1)                {$A$};
        \node[]    (2) [left of = 1]  {};
        \node[]    (3) [right of = 1] {};
        \node[box] (4) [below of = 1] {$B$};
        \node[]    (5) [left of = 4]  {};
        \node[]    (6) [right of = 4] {};
        \draw[-] (1) to (2);
        \draw[-] (1) to (3);
        \draw[-] (4) to (5);
        \draw[-] (4) to (6);
    \end{tikzpicture}
    \hfill
    \begin{tikzpicture}[baseline = {([yshift = -0.5ex] current bounding box.center)},
                        box/.style = {rectangle,
                                      draw = black,
                                      fill = blue!15!white,
                                      thick}]
        \node[box] (1)                {$u v^*$};
        \node[]    (2) [left of = 1]  {};
        \node[]    (3) [right of = 1] {};
        \draw[-] (1) to (2);
        \draw[-] (1) to (3);
    \end{tikzpicture}
    $=$
    \begin{tikzpicture}[baseline = {([yshift = -0.5ex] current bounding box.center)},
                        box/.style = {rectangle,
                                      draw = black,
                                      fill = blue!15!white,
                                      thick}]
        \node[box] (1)               {$u$};
        \node[]    (2) [left of = 1] {};
        \node[box] (3) [right of = 1] {$v$};
        \node[]    (4) [right of = 3] {};
        \draw[-] (1) to (2);
        \draw[-] (3) to (4);
    \end{tikzpicture}
    \hfill\null
\end{center}
The tensor contraction combines two diagrams 
by taking the sum over all common indices
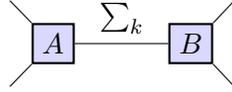
\begin{center}
    \begin{tikzpicture}[box/.style = {rectangle,
                                      draw = black,
                                      fill = blue!15!white,
                                      thick}]
        \node[box] (1)                              {$A$};
        \node[]    (2) [above left of = 1]          {};
        \node[]    (3) [below left of = 1]          {};
        \node[box] (4) [right of = 1, xshift = 5ex] {$B$};
        \node[]    (5) [above right of = 4]         {};
        \node[]    (6) [below right of = 4]         {};
        \draw[-] (1.north west) to                                 (2);
        \draw[-] (1.south west) to                                 (3);
        \draw[-] (1)            to node [midway, above] {$\sum_ k$} (4);
        \draw[-] (4.north east) to                                 (5);
        \draw[-] (4.south east) to                                 (6);
    \end{tikzpicture}
\end{center}
In particular, we recover the scalar product
\begin{center}
$\langle u, v\rangle = $
    \begin{tikzpicture}[baseline = {([yshift = -0.5ex] current bounding box.center)},
                        box/.style = {rectangle,
                                      draw = black,
                                      fill = blue!15!white,
                                      thick}]
        \node[box] (1)                {$u$};
        \node[box] (2) [right of = 1] {$v$};
        \draw[-] (1) to (2);
    \end{tikzpicture}
    $\: = \: \sum_k u_k \cdot v_k,$
\end{center}
the matrix product
\begin{center}
$(AB)_{ij} = $
    \begin{tikzpicture}[baseline = {([yshift = -1.5ex] current bounding box.center)},
                        box/.style = {rectangle,
                                      draw = black,
                                      fill = blue!15!white,
                                      thick}]
        \node[box] (1)                {$A$};
        \node[]    (2) [left of = 1]  {};
        \node[box] (3) [right of = 1] {$B$};
        \node[]    (4) [right of = 3] {};
        \draw[-] (1) to node [midway, above] {$i$} (2);
        \draw[-] (1) to (3);
        \draw[-] (3) to node [midway, above] {$j$} (4);
    \end{tikzpicture}
    $\: = \: \sum_k a_{ik} \cdot b_{kj},$
\end{center}
and the matrix trace
\begin{center}
$\Tr A = $
    \begin{tikzpicture}[baseline = {([yshift = -1.5ex] current bounding box.center)},
                        box/.style = {rectangle,
                                      draw = black,
                                      fill = blue!15!white,
                                      thick}]
        \node[box] (1) {$A$};
        \coordinate[left of = 1, xshift = 3ex]   (2);
        \coordinate[right of = 1, xshift = -3ex] (3);
        \coordinate[above of = 2, yshift = -3ex] (4);
        \coordinate[above of = 3, yshift = -3ex] (5);
        \draw[-] (1.west) to                   (2);
        \draw[-] (1.east) to                   (3);
        \draw[-] (2)      to [bend left = 90]  (4);
        \draw[-] (3)      to [bend right = 90] (5);
        \draw[-] (4)      to                   (5);
    \end{tikzpicture}
    $\: = \: \sum_k a_{kk}.$
\end{center}
Note that scalars multiply the diagrams and are depicted next to the diagram. We shall use the following three special tensors, of foremost importance in quantum information theory, which have wire-only diagrams: 
\begin{center}
    \hfill
    \begin{tikzpicture}[baseline = {([yshift = -0.5ex] current bounding box.center)}]
        \node[] (1)                {};
        \node[] (2) [right of = 1] {};
        \draw[-] (1) to (2);
    \end{tikzpicture}
    $= \: I$
    \hfill
    \begin{tikzpicture}[baseline = {([yshift = -0.5ex] current bounding box.center)}]
        \node[] (1) {$\frac{1}{\sqrt{d}}$};
        \coordinate[above of = 1, xshift = 3ex, yshift = -3ex] (2);
        \coordinate[below of = 1, xshift = 3ex, yshift = 3ex]  (3);
        \draw[-] (2) to [bend left = 90] (3);
    \end{tikzpicture}
    $= \: \ket{\Omega}$
    \hfill
    \begin{tikzpicture}[baseline = {([yshift = -0.5ex] current bounding box.center)}]
        \draw (0,0) circle (0.5);
    \end{tikzpicture}
    $= \: d.$
    \hfill\null
\end{center}
Finally, the transposition can be visually represented by swapping the ``input'' and the ``output'' wires of a box depicting a matrix. The partial transposition permutes just the wires of the corresponding spaces:
\begin{center}
    \begin{tikzpicture}[baseline = {([yshift = -0.5ex] current bounding box.center)},
                        box/.style = {rectangle,
                                      draw = black,
                                      fill = blue!15!white,
                                      thick}]
        \node[box] (1)                               {$A^\Gamma$};
        \node[]    (2) [left of = 1, yshift = 1ex]   {};
        \node[]    (3) [left of = 1, yshift = -1ex]  {};
        \node[]    (4) [right of = 1, yshift = 1ex]  {};
        \node[]    (5) [right of = 1, yshift = -1ex] {};
        \draw[-] ([yshift = 1ex]1.west)  to (2);
        \draw[-] ([yshift = -1ex]1.west) to (3);
        \draw[-] ([yshift = 1ex]1.east)  to (4);
        \draw[-] ([yshift = -1ex]1.east) to (5);
\end{tikzpicture}
$ = $
    \begin{tikzpicture}[baseline = {([yshift = -2ex] current bounding box.center)},
                        box/.style = {rectangle,
                                      draw = black,
                                      fill = blue!15!white,
                                      thick}]
        \node[box] (1)                                {$A$};
        \node[]    (4) [left of = 1, yshift = -1ex]   {};
        \node[]    (5) [right of = 1, yshift = -1ex]  {};
        \node[]    (8) [left of = 1, yshift = 4ex]   {};
        \node[]    (9) [right of = 1, yshift = 4ex]  {};
        \coordinate[left of = 1, xshift = 4ex, yshift = 1ex] (2);
        \coordinate[right of = 1, xshift = -4ex, yshift = 1ex] (3);
        \coordinate[above of = 2, yshift = -4.5ex] (6);
        \coordinate[above of = 3, yshift = -4.5ex] (7);
        \draw[-]                                 ([yshift = 1ex]1.west)  to                   (2);
        \draw[-]                                 ([yshift = 1ex]1.east)  to                   (3);
        \draw[-]                                 ([yshift = -1ex]1.west) to                   (4);
        \draw[-]                                 ([yshift = -1ex]1.east) to                   (5);
        \draw[-]                                 (2)                     to [bend left = 60]  (6);
        \draw[-]                                 (3)                     to [bend right = 60] (7);
        \draw[-]                                 (6)                     to [bend left = 15]  (9);
        \draw[-, line width = 2pt, draw = white] (7)                     to [bend right = 15] (8);
        \draw[-]                                 (7)                     to [bend right = 15] (8);
    \end{tikzpicture}
    $= (\transp \otimes \id ) A = \: A^{\Gamma}$
\end{center}

\section{Mathematical formulation of the cloning problem}\label{sec:math-clone}
According to Lemma \ref{SchurWeylLemma:1}, the partially transposed Choi matrix~$C^{\Gamma}_{\widetilde{T}}$ of a symmetrized quantum channel~$\widetilde{T} : \mathcal{M}_d \to \mathcal{M}_d \otimes \mathcal{M}_d$ is a linear combination of the~$6$ permutation operators of~$\mathfrak{S}_3$, which we represent graphically below
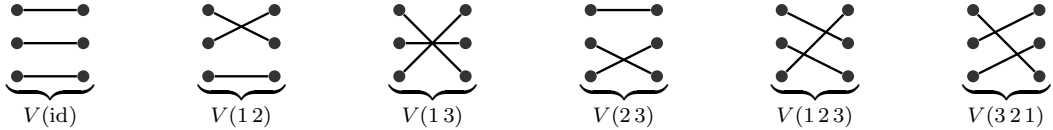
\begin{center}
    \hfill $\underbrace{\R{id}}_{V(\id)}$ \hfill $\underbrace{\R{12}}_{V(1 \, 2)}$ \hfill $\underbrace{\R{13}}_{V(1 \, 3)}$ \hfill $\underbrace{\R{23}}_{V(2 \, 3)}$ \hfill $\underbrace{\R{123}}_{V(1 \, 2 \, 3)}$ \hfill $\underbrace{\R{321}}_{V(3 \, 2 \, 1)}$ \hfill\null
\end{center}
Taking the partial transpose on the top~$\mathcal{M}_d$ space give us
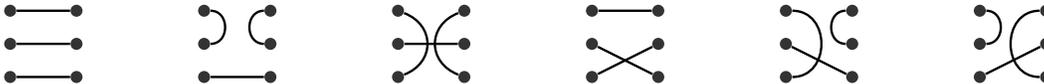
\begin{center}
    \hfill \RT{id} \hfill \RT{12} \hfill \RT{13} \hfill \RT{23} \hfill \RT{123} \hfill \RT{321} \hfill\null
\end{center}
and therefore the Choi matrix~$C_{\widetilde{T}}$ of an arbitrary~$1 \to 2$ symmetrized quantum channel can be written as
\begin{equation*}
    C_{\widetilde{T}} = \alpha_1 \cdot \smallRT{id} + \alpha_2 \cdot \smallRT{12} + \alpha_3 \cdot \smallRT{13} + \alpha_4 \cdot \smallRT{23} + \alpha_5 \cdot \smallRT{123} + \alpha_6 \cdot \smallRT{321}
\end{equation*}
for complex coefficients $\alpha_1, \ldots, \alpha_6 \in \mathbb C$. As a Choi matrix of a quantum channel,~$C_{\widetilde{T}}$ is positive semi-definite and satisfies~$\Tr_{\text{out}} \big{[} C_{\widetilde{T}} \big{]} = I$, where~$\Tr_{\text{out}}$ is the partial trace over the last~$\mathcal{D}_d$ spaces.

By taking the partial traces, we obtain the two marginals of~$\widetilde{T}$
\begin{align*}
    \widetilde{T}_1(\rho) = (\alpha_2 \, d + \alpha_5 + \alpha_6) \rho + (\alpha_1 \, d + \alpha_3 + \alpha_4) \frac{I}{d} \\
    \widetilde{T}_2(\rho) = (\alpha_3 \, d + \alpha_5 + \alpha_6) \rho + (\alpha_1 \, d + \alpha_2 + \alpha_4) \frac{I}{d}
\end{align*}
The partially transposed permutation operators~$\smallRT{123}$ and~$\smallRT{321}$ corresponding to the two cycles~$(1 \, 2 \, 3)$ and~$(3 \, 2 \, 1)$ of~$\mathfrak{S}_3$ contribute, with coefficients~$\alpha_5$ and~$\alpha_6$, to the figures of merit of both~$\widetilde{T}_1$ and~$\widetilde{T}_2$. However, we know by the no-cloning theorem that that a Choi matrix composed only of the two cycles of~$\mathfrak{S}_3$ is not the Choi matrix of a quantum channel, since the sum of $\smallRT{123}$ and $\smallRT{321}$ is not positive semi-definite. Hence, contributions from the other~$4$ operators are needed to ensure the positive semi-definiteness of the Choi matrix~$C_{\widetilde{T}}$. The~$1 \to 2$ asymmetric quantum cloning problems consist in finding the coefficients~${(\alpha_i)}_{i \in [6]}$ that make the Choi matrix~$C_{\widetilde{T}}$ satisfy the conditions of Choi's Theorem \ref{ChoiTheorem:1}. The trace preservation condition implies that
\begin{equation}\label{eq:trace-preservation}
    \alpha_1 \, d^2 + (\alpha_2 + \alpha_3 + \alpha_4) d + \alpha_5 + \alpha_6 = 1.
\end{equation}

We now turn to the problem of describing the fact that~$C_{\widetilde{T}}$ is positive semi-definite in terms of the coefficients~$\alpha_i$. We introduce, for all~$i$ in~$[d]$, the vectors~$u_i := \sqrt d \ket{\Omega}_{(1)} \otimes \ket{i}$ and~$v_i = \sqrt d \ket{\Omega}_{(2)} \otimes \ket{i}$, where~$\ket{\Omega}_{(k)}$ is the un-normalized maximally entangled state between the first and the~$k$-th spaces
\begin{align}
\label{eq:def-ui}    u_i &= \V{1} \\
\label{eq:def-vi}    v_i &= \V{2};
\end{align}
above, the horizontal line should be understood as having the index $i$ above it, selecting the $i$-th coordinate of the tensor slice. Then, the action of the partially transposed permutation operators on these vectors is
\begin{equation} \label{operatorsAction:1}
    \begin{aligned}
        \smallRT{id} \V{1} &= \V{1} & \smallRT{12} \V{1} &= d \cdot \V{1} & \smallRT{13} \V{1} &= \V{2} \\
        \smallRT{23} \V{1} &= \V{2} & \smallRT{123} \V{1} &= d \cdot \V{2} & \smallRT{321} \V{1} &= \V{1}
    \end{aligned}
\end{equation}
and
\begin{equation} \label{operatorsAction:2}
    \begin{aligned}
        \smallRT{id} \V{2} &= \V{2} & \smallRT{12} \V{2} &= \V{1} & \smallRT{13} \V{2} &= d \cdot  \V{2} \\
        \smallRT{23} \V{2} &= \V{1} & \smallRT{123} \V{2} &= \V{2} & \smallRT{321} \V{2} &= d \cdot \V{1}
    \end{aligned}
\end{equation}
The vectors~$u_i + v_i$ and~$u_i - v_i$ are eigenvectors of~$\smallRT{12} + \smallRT{13}$ and~$\smallRT{123} + \smallRT{321}$. Since~$\smallRT{12}$,~$\smallRT{13}$,~$\smallRT{123}$ and~$\smallRT{321}$ are rank-$d$ matrices, we have computed above their complete (nonzero) spectrums:
\begin{align*}
    \spec \bigg{(} \smallRT{12} + \smallRT{13} \bigg{)} &=
        \begin{cases}
            (d + 1)\: \times \: d \\
            (d - 1)\: \times \: d
        \end{cases} \\
    \spec \bigg{(} \smallRT{123} + \smallRT{321} \bigg{)} &=
        \begin{cases}
            \phantom{-}(d + 1)\: \times \: d \\
            -(d - 1)\: \times \: d.
        \end{cases}
\end{align*}
The other two operators, $\smallRT{id}$ and~$\smallRT{23}$, are unitary matrices and thus have full ($d^3$) rank.

In the next two sections we shall characterize the positive semi-definiteness of the matrix $C_{\tilde T}$ in terms of the coefficients $\alpha_i$, first by restricting to the first four, and then considering the general case. 

\section{Restricted quantum cloners}\label{sec:restricted}
We start by solving the~$1 \to 2$ asymmetric quantum cloning problems when the Choi matrix~$C_{\widetilde{T}}$ is a linear combination of only~$4$ partially transposed permutation operators of~$\mathfrak{S}_3$, i.e.
\begin{equation*}
    C_{\widetilde{T}} = \alpha_1 \cdot \smallRT{12} + \alpha_2 \cdot \smallRT{13} + \alpha_3 \cdot \smallRT{123} + \alpha_4 \cdot \smallRT{321}
\end{equation*}
Since we want~$C_{\widetilde{T}}$ to be positive semi-definite and in particular to be hermitian, we have the following relations:~$\alpha_1, \alpha_2 \in \mathbb{R}$ and~$\alpha_3 = \bar{\alpha}_4$. That is
\begin{equation*}
    C_{\widetilde{T}} = \alpha \cdot \smallRT{12} + \beta \cdot \smallRT{13} + \gamma \cdot \smallRT{123} + \bar{\gamma} \cdot \smallRT{321}
\end{equation*}
such that the partial trace condition from Eq.~\eqref{eq:trace-preservation} becomes~$d (\alpha + \beta) + 2 \Re(\gamma) = 1$.
Using equations (\ref{operatorsAction:1}) and (\ref{operatorsAction:2}) we proceed to the block diagonal decomposition of the~$4$ partially transposed permutation operators in the basis of the~$2 \, d$ vectors~$\V{1}$ and~$\V{2}$ from Eqs.~\eqref{eq:def-ui}-\eqref{eq:def-vi}, i.e.
\begin{equation*}
    {\smallRT{12}}_i =
        \begin{blockarray}{ccc}
            u_i & v_i \\
            \begin{block}{(cc)c}
                d & 1 & u_i\\
                0 & 0 & v_i\\
            \end{block}
        \end{blockarray}
\end{equation*}
and
\begin{align*}
    {\smallRT{13}}_i &=
    \begin{pmatrix}
        0 & 0 \\
        1 & d
    \end{pmatrix}
    &
    {\smallRT{123}}_i &=
    \begin{pmatrix}
        0 & 0 \\
        d & 1
    \end{pmatrix}
    &
    {\smallRT{321}}_i &=
    \begin{pmatrix}
        1 & d \\
        0 & 0
    \end{pmatrix}
\end{align*}
such that each partially transposed permutation operators is a direct sum of~$d$ such blocks and~$d^3-2d$ zero blocks, e.g.
\begin{equation*}
    \smallRT{12} = {\begin{pmatrix}  d & 1 \\ 0 & 0 \end{pmatrix}}^{\oplus \, d} \bigoplus \  {\bigzero \,}_{d^3-2d,d^3-2d}.
\end{equation*}
Then the block diagonal decomposition of~$C_{\widetilde{T}}$ becomes
\begin{equation} \label{blockDiagonalization:1}
    {\big{(} C_{\widetilde{T}} \big{)}}_i =
    \begin{pmatrix}
        \alpha \, d + \bar{\gamma} & \alpha + \bar{\gamma} \, d \\
        \beta + \gamma \, d & \beta \, d + \gamma
    \end{pmatrix}
\end{equation}

It is well known that a~$2 \times 2$ hermitian matrix~$M$ is positive semi-definite if and only if~$\det M \geq 0$ and~$\Tr M \geq 0$. Then~$C_{\widetilde{T}}$ is positive semi-definite if and only if each of its~$2 \times 2$ blocks of equation (\ref{blockDiagonalization:1}) are is positive semi-definite. That is
\begin{align*}
    \Tr \Big{[} {\big{(} C_{\widetilde{T}} \big{)}}_i \Big{]} &= d (\alpha + \beta) + 2  \Re(\gamma) \geq 0
\end{align*}
and
\begin{align*}
    \det \Big{[} {\big{(} C_{\widetilde{T}} \big{)}}_i \Big{]} &= \alpha \, \beta - {|\gamma|}^2  \geq 0.
\end{align*}
The first condition is always true since we must have~$d (\alpha + \beta) + 2 \Re(\gamma) = 1$. Finally, the Choi matrix~$C_{\widetilde{T}}$ is the Choi matrix of a quantum channel, and thus a quantum cloning map, when both
\begin{equation} \label{CloningMapCondition:1}
    d (\alpha + \beta) + 2  \Re(\gamma) = 1 \qquad \text{and} \qquad \alpha \, \beta \geq {|\gamma|}^2.
\end{equation}

The two marginals of~$\widetilde{T}$ are
\begin{align*}
    \widetilde{T}_1(\rho) = \underbrace{\big{(} \alpha \, d + 2  \Re(\gamma) \big{)}}_{p_1} \rho + \beta  I \\
    \widetilde{T}_2(\rho) = \underbrace{\big{(} \beta \, d + 2  \Re(\gamma) \big{)}}_{p_2} \rho + \alpha I,
\end{align*}
with~$p_1$ and~$p_2$ being the two figures of merit. With these definitions, the marginals $\widetilde T_i$ satisfy $\widetilde T_i(\rho)=p_i\rho+\left(1-p_i\right)\frac{I_d}{d}$. Hence, the equations (\ref{CloningMapCondition:1}) lead us to the main result of this section, describing the admissible region of singlet fractions $(p_1,p_2)$ for the quantum cloning problem, when the cloner is restricted to the first four permutation operators. 
\begin{theorem}\label{thm:restricted}
    The admissible region of the singlet fractions~$p_1$ and~$p_2$ for the \emph{restricted} universal~$1 \to 2$ asymmetric quantum cloning is the \emph{ellipse} given by:
    \begin{equation*}
      \Bigg \{(p_1, p_2) \in \left[\frac{-1}{d^2-1}, 1\right] \, : \,  \frac{(1 - p_1)(1 - p_2)}{d^2} \geq \bigg( \frac{p_1 + p_2 - 1}{2} \bigg) ^2 \Bigg\}.
      \end{equation*} 
      Equivalently, introducing the new variables $(s,t)=(p_1 + p_2,p_1 - p_2),a=(d^2-1)^{-1/2},b=(d^2-1)^{-1}$, the figures of merit for the the optimal quantum cloners are given by the ellipse
$$ \frac{t^2}{a^2}  +\frac{\left(s - \frac{d^2-2}{d^2-1} \right)^2}{b^2} \leq d^2.$$
    
\end{theorem}
\begin{proof}
	A pair $(p_1, p_2) \in  \big[\frac{-1}{d^2-1}, 1\big]$ is feasible if and only if there exists coefficients $\alpha, \beta \in \mathbb R$ and $\gamma \in \mathbb C$ satisfying Eq.~\eqref{CloningMapCondition:1}, with
	$$p_1  = \alpha d + 2 \Re(\gamma) = 1-\beta d  \qquad \text{ and } \qquad p_2 = \beta d + 2 \Re(\gamma) = 1- \alpha d.$$
	A complex $\gamma$ satisfying both equations from \eqref{CloningMapCondition:1} exists if and only if 
	$$\left(\frac{1-d(\alpha+\beta)}{2} \right)^2 \leq \alpha \beta.$$
Rewriting the above inequality in terms of $p_{1,2}$ yields the equation in the statement. 
\end{proof}

	Note that the above result also contains the admissible values for \emph{negative} singlet fractions, a regime which we study for the first time. 
	
	The origin $(0,0)$ does not belong the admissible set for $d \geq 3$; this is due to the restriction on the form of the Choi matrix we consider. The origin corresponds in particular to the totally depolarizing channel, which is ruled out since the identity matrix is not allowed in the current setting. 

Finally, note that, for any $d$, the largest symmetric admissible point $(p,p)$ is given by
$$p_{\max} = \frac{2+d}{2(1+d)},$$
which is the value obtained in \cite{werner1998optimal}. In particular, the optimal symmetric cloning channels can be chosen from the restricted family of cloners considered in this section. 


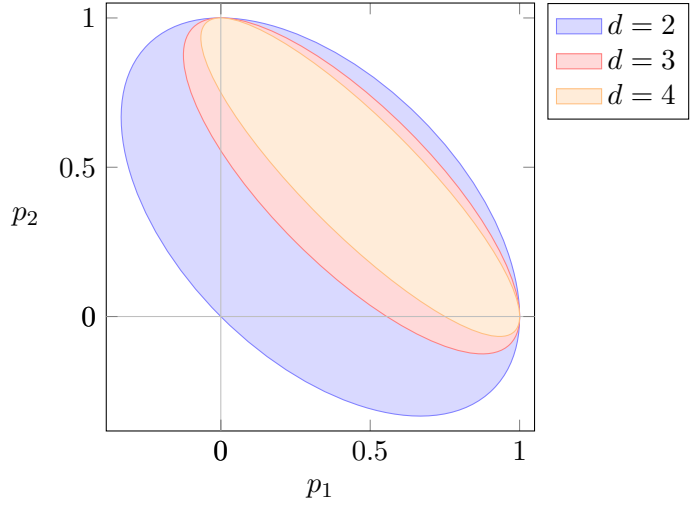
\begin{figure}[H]
    \centering
    \begin{tikzpicture}
        \begin{axis}[view = {0} {90},
                     xlabel = $p_1$,
                     ylabel = $p_2$,
                     ylabel style = {rotate = -90},
                     xmin =  {-1 / (2^2 - 1) - 0.05},
                     xmax = {1 + 0.05},
                     ymin = {-1 / (2^2 - 1) - 0.05},
                     ymax = {1 + 0.05},
                     axis equal image,
                     axis on top = true,
                     extra x ticks = 0,
    	             extra y ticks = 0,
    	             extra tick style = {grid = major},
    	             legend pos = outer north east]
            \begin{scope}[rotate around = {-45:(0,0)}]
                \draw[draw = blue!50!white,
                      fill = blue!15!white]
                    (0,{sqrt(2) / 3})
                    ellipse [x radius = {sqrt(2 / 3)}, y radius = {sqrt(2) / 3}];
                \draw[draw = red!50!white,
                      fill = red!15!white]
                    (0,{7 / (8 * sqrt(2))})
                    ellipse [x radius = {3 / 4}, y radius = {3 / (8 * sqrt(2))}];
                \draw[draw = orange!50!white,
                      fill = orange!15!white]
                    (0,{(7 * sqrt(2)) / 15})
                    ellipse [x radius = {2 * sqrt(2 / 15)}, y radius = {(2 * sqrt(2)) / 15}];
            \end{scope}
            \addlegendimage{area legend, draw = blue!50!white, fill = blue!15!white}
            \addlegendentry{$d = 2$}
            \addlegendimage{area legend, draw = red!50!white, fill = red!15!white}
            \addlegendentry{$d = 3$}
            \addlegendimage{area legend, draw = orange!50!white, fill = orange!15!white}
            \addlegendentry{$d = 4$}
        \end{axis}
    \end{tikzpicture}
    \caption{Admissible regions for the figures of merit with~$4$ permutation operators.}
    \label{FiguresMerit:1}
\end{figure}

\section{General quantum cloners}\label{sec:general}
We consider in this section the general case of $1 \to 2$ quantum cloners, and compute the form of the admissible region for the singlet fractions of the marginals. The complete description of a Choi matrix~$C_{\widetilde{T}}$ for the~$1 \to 2$ asymmetric quantum cloning problems is given the 6 partially transposed permutation operators of~$\mathfrak{S}_3$, i.e.
\begin{equation*}
    C_{\widetilde{T}} = \alpha \cdot \smallRT{12} + \beta \cdot \smallRT{13} + \gamma \cdot \smallRT{123} + \bar{\gamma} \cdot \smallRT{321} + \varepsilon_1 \cdot \smallRT{id} + \varepsilon_2 \cdot \smallRT{23}.
\end{equation*}
The hermitian condition on~$C_{\widetilde{T}}$ imposes that:~$\alpha, \beta,\varepsilon_1, \varepsilon_2 \in \mathbb{R}$, and the partial trace condition reads~$d (\alpha + \beta) + 2  \Re(\gamma) + d^2 \varepsilon_1 + d \varepsilon_2 = 1$. As in the previous section, the action of $C_{\widetilde{T}}$ decomposes on the $d$ subspaces spanned by the vectors $u_i,v_i$ from Eqs.~\eqref{eq:def-ui}-\eqref{eq:def-vi}. Let~$\boldsymbol{V} = \Span \bigg{\{} \V{1} , \V{2} \bigg{\}}$, then~$\boldsymbol{V} \subset \mathcal{H} \otimes \vee^2(\mathcal{H})$, where~$\vee^2(\mathcal{H})$ is the \emph{symmetric subspace} on~$\mathcal{H}^{\otimes 2}$. The subspace~$\boldsymbol{V}$ is invariant by the two operators $\smallRT{id}$ and~$\smallRT{23}$ that have both full ($d^3$) rank. On~$\boldsymbol{V}^\perp$, the spectrum of~$\varepsilon_1 \cdot \smallRT{id} + \varepsilon_2 \cdot \smallRT{23}$ is
\begin{equation*}
    \spec {\raisebox{-0.3em}{\bigg{|}}}_{\boldsymbol{V}^\perp} \bigg{(} \varepsilon_1 \cdot \smallRT{id} + \varepsilon_2 \cdot \smallRT{23} \bigg{)} =
        \begin{cases}
            \varepsilon_1 + \varepsilon_2 &\times \: d \, \frac{d (d + 1)}{2} - 2 d \\
            \varepsilon_1 - \varepsilon_2 &\times \: d \, \frac{d (d - 1)}{2}
        \end{cases}.
\end{equation*}
Then the complete block diagonal decomposition of~$C_{\widetilde{T}}$ is made of~$2 \times 2$ and~$1 \times 1$ blocks. The~$1 \times 1$ blocks~$(\varepsilon_1 + \varepsilon_2)$ and~$(\varepsilon_1 - \varepsilon_2)$ are positive when~$\varepsilon_1 \geq |\varepsilon_2|$. On~$\boldsymbol{V}$ we have the block diagonal decomposition
\begin{equation*} 
    {\big{(} C_{\widetilde{T}} \big{)}}_i =
    \begin{pmatrix}
        \alpha \, d + \bar{\gamma} + \varepsilon_1 & \alpha + \bar{\gamma} \, d + \varepsilon_2 \\
        \beta + \gamma \, d + \varepsilon_2 & \beta \, d + \gamma + \varepsilon_1
    \end{pmatrix}.
\end{equation*}
After the change of basis
\begin{align*}
    \V{1} &\longmapsto \frac{1}{\sqrt{2 (d + 1)}} \bigg{(} \V{1} + \V{2} \bigg{)} \\
    \V{2} &\longmapsto \frac{1}{\sqrt{2 (d - 1)}} \bigg{(} \V{1} - \V{2} \bigg{)}
\end{align*}
we obtain the hermitian block diagonal decomposition of~$C_{\widetilde{T}}$
\begin{equation*}
    {\big{(} C_{\widetilde{T}} \big{)}}_i =
    \begin{pmatrix}
        (d + 1) \frac{\alpha + \beta + 2 \cdot \Re(\gamma)}{2} + \varepsilon_1 + \varepsilon_2 & \frac{\sqrt{d^2 - 1}}{2} (\alpha - \beta) \\[1em]
        \frac{\sqrt{d^2 - 1}}{2} (\alpha - \beta) & (d + 1) \frac{\alpha + \beta - 2 \cdot \Re(\gamma)}{2} + \varepsilon_1 - \varepsilon_2
    \end{pmatrix}.
\end{equation*}

The two figures of merit of~$\widetilde{T}$ are~$p_1 = \alpha \, d + 2  \Re(\gamma)$ and~$p_2 = \beta \, d + 2 \Re(\gamma)$. By setting~$s: = p_1 + p_2$,~$t := p_1 - p_2$ together with the relation~$d (\alpha + \beta) + 2  \Re(\gamma) + d^2 \varepsilon_1 + d \varepsilon_2 = 1$ we have
\begin{align*}
    {\big{(} C_{\widetilde{T}} \big{)}}_i &=
    \begin{pmatrix}
        \frac{(d + 1)}{2 d} s + \frac{(d + 1) (d - 2)}{d} \Re(\gamma) + \varepsilon_1 + \varepsilon_2 & \frac{\sqrt{d^2 - 1}}{2 d} t \\[1em]
       \frac{\sqrt{d^2 - 1}}{2 d} t & \frac{(d + 1)}{2 d} s - \frac{(d - 1) (d + 2)}{d} \Re(\gamma) + \varepsilon_1 - \varepsilon_2
    \end{pmatrix} \\[1em]
    &=
    \begin{pmatrix}
        \frac{(d^2 - 1) s + (d - 1) d \big{(} (d^2 - 2) \varepsilon_1 + \varepsilon_2 \, d - 1 \big{)} + 2}{2 d} & \frac{\sqrt{d^2 - 1}}{2 d} t \\[1em]
        \frac{\sqrt{d^2 - 1}}{2 d} t & - \frac{(d^2 - 1) s + (d + 1) d \big{(} (d^2 - 2) \varepsilon_1 + \varepsilon_2 \, d - 1 \big{)} + 2}{2 d}
    \end{pmatrix}.
\end{align*}
Let~$\lambda := - d \big{(} (d^2 - 2) \varepsilon_1 + \varepsilon_2 \, d - 1 \big{)}$, then the hermitian block diagonal decomposition of~$C_{\widetilde{T}}$ reduces to 
\begin{equation} \label{blockDiagonalization:2}
    {\big{(} C_{\widetilde{T}} \big{)}}_i = \frac{1}{2 d}
    \begin{pmatrix}
        (d^2 - 1) s - (d - 1) \lambda + 2 & \sqrt{d^2 - 1} \, t \\[1em]
        \sqrt{d^2 - 1} \, t & - \big{(} (d^2 - 1) s - (d + 1) \lambda + 2 \big{)}
    \end{pmatrix}.
\end{equation}
In this way, the positivity condition on each of the~$2 \times 2$ blocks of equation (\ref{blockDiagonalization:2}) becomes
\begin{equation} \label{ellipse:1}
    0 \leq \lambda \leq d \qquad \text{and} \qquad \frac{t^2}{a^2} + \frac{(s - c)^2}{b^2} \leq \lambda^2,
\end{equation}
with~$a = \frac{1}{\sqrt{d^2 - 1}}$,~$b = \frac{1}{d^2 - 1}$,~$c = \frac{\lambda d - 2}{d^2 - 1}$, where the last condition is the equation of a \emph{shifted ellipse}, see Figure \ref{FiguresMerit:2}. In conclusion, we have just proven the main result of this section, a characterization of the admissible region of figures of merit for the $1 \to 2$ asymmetric quantum cloning problem.

\begin{theorem}\label{thm:unrestricted}
    The admissible region of the singlet fractions for the universal~$1 \to 2$ asymmetric quantum cloning is the union of a family of ellipses indexed by $\lambda \in [0,d]$ given by \begin{equation}
    \qquad \frac{t^2}{a^2} + \frac{(s - c_\lambda)^2}{b^2} \leq \lambda^2,
\end{equation}
with~$a = \frac{1}{\sqrt{d^2 - 1}}$,~$b = \frac{1}{d^2 - 1}$,~$c_\lambda = \frac{\lambda d - 2}{d^2 - 1}$.
\end{theorem}

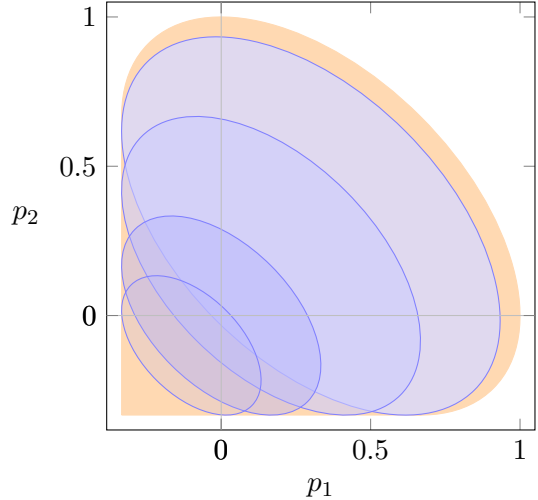
\begin{figure}[H]
    \centering
    \begin{tikzpicture}
        \begin{axis}[view = {0} {90},
                     xlabel = $p_1$,
                     ylabel = $p_2$,
                     ylabel style = {rotate = -90},
                     xmin =  {-1 / (2^2 - 1) - 0.05},
                     xmax = {1 + 0.05},
                     ymin = {-1 / (2^2 - 1) - 0.05},
                     ymax = {1 + 0.05},
                     axis equal image,
                     axis on top = true,
                     extra x ticks = 0,
    	             extra y ticks = 0,
    	             extra tick style = {grid = major}]
    	       \begin{scope}[rotate around = {-45:(current axis.origin)}]
                    \draw[fill,orange!30!white] (0,{sqrt(2) / 3})
                        ellipse [x radius = {sqrt(2 / 3)},y radius = {sqrt(2) / 3}];
                \end{scope}
                \draw[fill,orange!30!white] ({-1 / 3},{-1 / 3}) -- ({-1 / 3},{2 / 3}) -- ({2 / 3},{-1 / 3})-- ({-1 / 3},{-1 / 3});
            \begin{scope}[rotate around = {-45:(0,0)}]
                \draw[draw = blue!50!white,
                      fill = blue!15!white,
                      fill opacity = 0.8]
                    (0,0.424264)
                    ellipse [x radius = 0.775672, y radius = 0.447834];
                \draw[draw = blue!50!white,
                      fill = blue!20!white,
                      fill opacity = 0.6]
                    (0,0.235702)
                    ellipse [x radius = 0.612372, y radius = 0.353553];
                \draw[draw = blue!50!white,
                      fill = blue!30!white,
                      fill opacity = 0.4]
                    (0,0)
                    ellipse [x radius = 0.408248, y radius = 0.235702];
                \draw[draw = blue!50!white,
                      fill = blue!30!white,
                      fill opacity = 0.2]
                    (0,-0.141421)
                    ellipse [x radius = 0.285774, y radius = 00.164992];
            \end{scope}
        \end{axis}
    \end{tikzpicture}
    \caption{The admissible region \begin{tikzpicture} \fill[orange!30!white] (0,0)rectangle(1.5ex,1.5ex); \end{tikzpicture} of the figures of merit of a general $1 \to 2$ asymmetric cloner is the union of a continuous family of ellipses \begin{tikzpicture} \fill[blue!30!white] (0,0)rectangle(1.5ex,1.5ex); \end{tikzpicture} (four ellipses shown, with~$\varepsilon_1 = \varepsilon_2$ and~$d = 2$).}
    \label{FiguresMerit:2}
\end{figure}

Let us point out that, in the general setting of this section, we can reach the origin $(0,0)$ for any dimension, by varying the parameters $\varepsilon_{1,2}$. However, the \emph{optimal} quantum cloners (i.e.~the points in the upper-right part of Figure \ref{FiguresMerit:2} are the same as the ones from Theorem \ref{thm:restricted}: optimal quantum cloners only require the 4 permutation matrices discussed in Section \ref{sec:restricted}. This is to be expected, since the fifth and the sixth permutation operators we consider in this section do not contribute to increasing the figures of merit. The coefficients $\varepsilon_{1,2}$ allow to sweep the lower-left part of the admissible region, away from the optimal curve on the top-right. 

\bigskip
\noindent\textit{Acknowledgments.} This research was supported by the ANR project ``\href{https://esquisses.math.cnrs.fr/}{ESQuisses}'' (ANR-20-CE47-0014-01). This research was supported by the ANR Program`Investissements d'Avenir" with reference ANR-11-LABX-0040 trough the Labex CIMI. The research of C.P.\ has been supported by project  ``\href{https://qtraj.math.cnrs.fr/}{QTraj}''(ANR-20-CE40-0024-01) of the French National Research Agency (ANR)

\bibliographystyle{alpha}
\bibliography{bibliography}

\begin{thebibliography}{S{\'C}HM14}

\bibitem[BC17]{bridgeman2017hand}
Jacob~C Bridgeman and Christopher~T Chubb.
\newblock Hand-waving and interpretive dance: an introductory course on tensor
  networks.
\newblock {\em Journal of Physics A: Mathematical and Theoretical},
  50(22):223001, 2017.

\bibitem[BH96]{buvzek1996quantum}
Vladimir Bu{\v{z}}ek and Mark Hillery.
\newblock Quantum copying: Beyond the no-cloning theorem.
\newblock {\em Physical Review A}, 54(3):1844, 1996.

\bibitem[Cer98]{cerf1998asymmetric}
Nicolas~J Cerf.
\newblock Asymmetric quantum cloning machines.
\newblock {\em J. Mod. Opt.}, 47(KRL-MAP-224):187, 1998.

\bibitem[Cer00]{cerf2000asymmetric}
Nicolas~J Cerf.
\newblock Asymmetric quantum cloning in any dimension.
\newblock {\em Journal of modern optics}, 47(2-3):187--209, 2000.

\bibitem[Cho75]{choi1975completely}
Man-Duen Choi.
\newblock Completely positive linear maps on complex matrices.
\newblock {\em Linear algebra and its applications}, 10(3):285--290, 1975.

\bibitem[{\'C}HS12]{cwiklinski2012region}
Piotr {\'C}wikli{\'n}ski, Micha{\l} Horodecki, and Micha{\l} Studzi{\'n}ski.
\newblock Region of fidelities for a $1 \to n$ universal qubit quantum cloner.
\newblock {\em Physics Letters A}, 376(32):2178--2187, 2012.

\bibitem[CK17]{coecke2017picturing}
Bob Coecke and Aleks Kissinger.
\newblock {\em Picturing quantum processes}.
\newblock Cambridge University Press, 2017.

\bibitem[EW01]{eggeling2001separability}
Tilo Eggeling and Reinhard~F Werner.
\newblock Separability properties of tripartite states with $u \otimes u
  \otimes u$ symmetry.
\newblock {\em Physical Review A}, 63(4):042111, 2001.

\bibitem[FFC05]{fiuravsek2005highly}
Jarom{\'\i}r Fiur{\'a}{\v{s}}ek, Radim Filip, and Nicolas~J Cerf.
\newblock Highly asymmetric quantum cloning in arbitrary dimension.
\newblock {\em Quantum Information \& Computation}, 5(7):583--592, 2005.

\bibitem[Has17]{hashagen2016universal}
A-L Hashagen.
\newblock Universal asymmetric quantum cloning revisited.
\newblock {\em Quantum Information \& Computation}, 17(9-10):0747--0778, 2017.

\bibitem[HLA21]{hsieh2021quantum}
Chung-Yun Hsieh, Matteo Lostaglio, and Antonio Acín.
\newblock Quantum channel marginal problem.
\newblock {\em arXiv preprint arXiv:2102.10926}, 2021.

\bibitem[HMZ16]{heinosaari2016invitation}
Teiko Heinosaari, Takayuki Miyadera, and M{\'a}rio Ziman.
\newblock An invitation to quantum incompatibility.
\newblock {\em Journal of Physics A: Mathematical and Theoretical},
  49(12):123001, 2016.

\bibitem[Kay14]{kay2014optimal}
Alastair Kay.
\newblock Optimal universal quantum cloning: Asymmetries and fidelity measures.
\newblock {\em arXiv preprint arXiv:1407.4951}, 2014.

\bibitem[Kay16]{kay2016optimal}
Alastair Kay.
\newblock Optimal universal quantum cloning: Asymmetries and fidelity measures.
\newblock {\em Quantum Information and Computation}, 16(11 \& 12):0991--1028,
  2016.

\bibitem[KKR09]{kay2009optimal}
Alastair Kay, Dagomir Kaszlikowski, and Ravishankar Ramanathan.
\newblock Optimal cloning and singlet monogamy.
\newblock {\em Physical review letters}, 103(5):050501, 2009.

\bibitem[KRK12]{kay2012optimal}
Alastair Kay, Ravishankar Ramanathan, and Dagomir Kaszlikowski.
\newblock Optimal asymmetric quantum cloning.
\newblock {\em arXiv preprint arXiv:1208.5574}, 2012.

\bibitem[KW99]{keyl1999optimal}
Michael Keyl and Reinhard~F Werner.
\newblock Optimal cloning of pure states, testing single clones.
\newblock {\em Journal of Mathematical Physics}, 40(7):3283--3299, 1999.

\bibitem[NC10]{nielsen2010quantum}
Michael~A Nielsen and Isaac~L Chuang.
\newblock {\em Quantum computation and quantum information}.
\newblock Cambridge University Press, 2010.

\bibitem[Pen71]{penrose1971applications}
Roger Penrose.
\newblock Applications of negative dimensional tensors.
\newblock {\em Combinatorial mathematics and its applications}, 1:221--244,
  1971.

\bibitem[S{\'C}HM14]{studzinski2014group}
Micha{\l} Studzi{\'n}ski, Piotr {\'C}wikli{\'n}ski, Micha{\l} Horodecki, and
  Marek Mozrzymas.
\newblock Group-representation approach to $1 \to n$ universal quantum cloning
  machines.
\newblock {\em Physical Review A}, 89(5):052322, 2014.

\bibitem[SHM13]{studzinski2013commutant}
Micha{\l} Studzi{\'n}ski, Micha{\l} Horodecki, and Marek Mozrzymas.
\newblock Commutant structure of $u^{\otimes (n- 1)} \otimes u^*$
  transformations.
\newblock {\em Journal of Physics A: Mathematical and Theoretical},
  46(39):395303, 2013.

\bibitem[VW01]{vollbrecht2001entanglement}
Karl Gerd~H Vollbrecht and Reinhard~F Werner.
\newblock Entanglement measures under symmetry.
\newblock {\em Physical Review A}, 64(6):062307, 2001.

\bibitem[Wat18]{watrous2018theory}
John Watrous.
\newblock {\em The theory of quantum information}.
\newblock Cambridge University Press, 2018.

\bibitem[WBC15]{wood2015tensor}
Christopher~J Wood, Jacob~D Biamonte, and David~G Cory.
\newblock Tensor networks and graphical calculus for open quantum systems.
\newblock {\em Quantum Information \& Computation}, 15(9-10):759--811, 2015.

\bibitem[Wer98]{werner1998optimal}
Reinhard~F Werner.
\newblock Optimal cloning of pure states.
\newblock {\em Physical Review A}, 58(3):1827, 1998.

\bibitem[Wey16]{weyl2016classical}
Hermann Weyl.
\newblock {\em The classical groups}.
\newblock Princeton university press, 2016.

\bibitem[WZ82]{wootters1982single}
William~K Wootters and Wojciech~H Zurek.
\newblock A single quantum cannot be cloned.
\newblock {\em Nature}, 299(5886):802--803, 1982.

\end{thebibliography}

\end{document}